\documentclass[conf]{new-aiaa}
\usepackage[utf8]{inputenc}

\usepackage[fancythm,fancybb]{jphmacros2e} 
\usepackage{subcaption}
\usepackage{graphicx}
\usepackage{tikz}
\usepackage{tikz-3dplot}
\usepackage{amsmath}
\usepackage[version=4]{mhchem}
\usepackage{siunitx}
\usepackage{longtable,tabularx}
\newcommand{\norm}[1]{\left\lVert#1\right\rVert}
\usepackage{comment}

\newcommand{\wlen}{\Lambda}

\title{Optimal Satellite Maneuvers for Spaceborne Jamming Attacks}

\author{Filippos Fotiadis\footnote{Postdoctoral Researcher, Oden Institute for Computational Engineering \& Sciences, e-mail: ffotiadis@utexas.edu}, Quentin Rommel\footnote{Graduate Research Assistant, Department of Aerospace Engineering \& Engineering Mechanics, e-mail: quentin.rommel@utmail.utexas.edu}, Brian M. Sadler\footnote{Senior Research Fellow, Oden Institute for Computational Engineering \& Sciences, e-mail: brian.sadler@ieee.org}, Ufuk Topcu\footnote{Professor, Department of Aerospace Engineering \& Engineering Mechanics, e-mail: utopcu@utexas.edu}}
\affil{The University of Texas at Austin, Austin, TX, 78712}

\begin{document}

\maketitle

\begin{abstract}
Satellites are becoming exceedingly critical for communication, making them prime targets for cyber-physical attacks. We consider a rogue satellite in low Earth orbit that jams the uplink communication between another satellite and a ground station. To achieve maximal interference with minimal fuel consumption, the jammer carefully maneuvers itself relative to the target satellite's antenna. We cast this maneuvering objective as a two-stage optimal control problem, involving i) repositioning to an efficient jamming position before uplink communication commences; and ii) maintaining an efficient jamming position after communication has started. We obtain the optimal maneuvering trajectories for the jammer and perform simulations to show how they enable the disruption of uplink communication with reasonable fuel consumption.
\end{abstract}




\section{Introduction}
\vspace{3mm}

Space systems operate in a hostile environment where debris, solar radiation, and adversaries can compromise their integrity. Adversaries, in particular, are an emerging attack vector against them owing to the strategic planning of spacefaring nations and the growing accessibility of space infrastructure to
non-state entities \cite{cubesat, securityinspace}. A notable example is the 2022 cyberattack launched against the 
satellite operator Viasat, which left several thousand internet users across Europe with no internet access \cite{viasat}. At the same time, various other adversarial attacks take place but remain largely unreported \cite{ascend2}.
Ensuring the secure operation of space systems thus requires that we understand every avenue an adversary may use to launch an attack against them \cite{survey}.

Within the context of space systems, satellites are particularly central in communications and are thus a target of \textit{jamming attacks} \cite{jamming1, jamming2, jamming3, jamming4}.  Jamming takes place when a friendly communication signal arriving at a particular antenna is corrupted by electromagnetic interference sent by an adversary from a different source point. Specifically for satellites, the most common sources of jamming are adversaries stationed on the ground since ground-based jamming attacks are cheap to launch and require no access to a space system. Such attacks are well-documented both in prior real-world events \cite{ruegamer2015jamming, event1} and in the literature on space cybersecurity \cite{groundjam1, groundjam2, groundjam3, groundjam4}. However, cybersecurity research has paid less attention to the potential of space-based jamming attacks, which are becoming more feasible due to lower barriers in space access. These attacks pose a graver threat in uplink signal disruption owing to fewer line-of-sight constraints and smaller attenuation in jamming power, and in downlink signal disruption owing to easier access to remote areas. 

Research on space security focuses on both identifying potential threats and proposing security measures to compensate for them. For instance, \cite{ascend} describes how a rogue satellite may maneuver to launch sensor or actuator jamming attacks towards other space systems, and proposes proactive and reactive approaches to defend against these attacks. A particularly desirable defense mechanism is that of attack detection \cite{satdet1, satdet2, satdet3}, as it allows one to switch into mitigation mode once adversarial interference is identified. In \cite{satdet3}, the authors note that developing correct detection mechanisms becomes complicated when the adversary is performing a maneuver to maximize jamming efficiency. Yet, little is understood about the design and feasibility of such a maneuver, particularly given the limited fuel resources that exist in space.
Moreover, while prior work has extensively investigated adversarial maneuvers in the context of pursuit-evasion \cite{mehlman2024cat, evasion1, evasion2, evasion3, evasion4}, there has been less attention to designing and understanding maneuvers for optimal jamming.

We address this largely unexplored problem by showing how a jammer satellite can maneuver in low Earth orbit to maximize jamming impact on a defender satellite, while using only minimal fuel resources. Our contributions are threefold:

\begin{itemize}
\item \textbf{Jamming impact quantification as a function of intersatellite position:} We show that the impact of the jammer's interference depends solely on the relative position of the jammer with respect to the defender. We make this relation explicit by capturing jamming impact through the signal-to-interference-plus-noise ratio (SINR).

\item \textbf{Orbital mechanics framework for fuel-efficient jamming:} Using our jamming impact quantification, we formulate a two-stage maneuvering problem in which the jammer needs to position itself relative to the defender to maximize interference. The first stage of the maneuver repositions the jammer \textit{before} the communication window to minimize SINR at the window's start while minimizing fuel consumption. The second one keeps the jammer at an efficient jamming position \textit{during} the communication window while keeping fuel consumption low.

\item \textbf{Derivation of orbital maneuvers using optimal control:} We show that computing the first orbital maneuver boils down to solving an algebraic equation. Moreover, we derive a boundary value problem whose numerical solution yields the second, station-keeping maneuver.
\end{itemize}

We empirically validate our findings in a simulation of two low Earth orbit satellites, showing how they allow a jammer to disrupt uplink communication without engaging in a fuel-expensive pursuit of the defender. Our work is the first to show how this is achievable by combining communication physics, orbital mechanics, and optimal control theory.

\section{Problem Formulation}\label{sec:pr}
\vspace{3mm}

\subsection{Jamming Impact as a Function of Intersatellite Position}
\vspace{3mm}

Consider a satellite, called the defender, that is on a circular low Earth orbit and uses an antenna to receive a friendly communication signal from a ground station. 
Consider also an adversarial satellite in the vicinity of the defender, called the jammer, whose purpose is to disrupt the defender's communication with the ground station using electromagnetic interference, i.e., jamming. We consider the following assumptions.

\begin{assumption}\label{ass:1}
The jammer's and the ground station's antennas constantly point to the defender.
\end{assumption}
\begin{assumption}\label{ass:2}
    The defender's antenna constantly points to the Earth's center.
\end{assumption}

Denote the friendly communication signal that the ground station sends to the defender as $s(t)\in\mathbb{C}$, and the interference signal that the jammer satellite targets towards the defender as $a(t)\in\mathbb{C}$. Under Assumptions \ref{ass:1}-\ref{ass:2}, the received signal plus jamming at the defender is
\begin{equation}\label{eq:rec}
y(t)=h_s(t)s(t)+h_a(t)a(t)+\eta(t),~t\ge0,
\end{equation}
where $\eta(t)\in\mathbb{C}$ is measurement noise following the distribution $\mathcal{N}(0,\sigma_{\eta}^2)$, and $h_s(t)\in\mathbb{C}, h_a(t)\in\mathbb{C}$ are the
receiving pattern functions of the friendly and jamming signals on the defender.
Note that, since the defender and the jammer satellites are in orbital motion, the receiving pattern functions $h_s(t), h_a(t)$
are constantly changing. This is because the powers of the friendly and the jamming signal on the defender's antenna vary in time, depending on the traveled distance as well as the angle of signal reception.

To quantify the effect of the jammer's interference on the defender's antenna, we use 
the signal-to-interference-plus-noise ratio (SINR) metric. This follows the formula
\begin{equation}\label{eq:SINR}
\mathrm{SINR}(t)=\frac{P_sH_s(t)}{P_aH_a(t)+\sigma_{\eta}^2},
\end{equation}
where $P_s, P_a$ are the average powers of $s(t)$ and $a(t)$, and $H_s(t)$ and $H_a(t)$ are functions that contain information regarding the sender's and jammer's antenna gains, the defender's antenna gains, and free-space path losses. Specifically for the jammer, under Assumptions \ref{ass:1}-\ref{ass:2}, it follows that
\begin{equation}\label{eq:H}
\begin{split}
H_a(t)=H_a(p(t))=G_aG_d(\theta(p(t)))L(\norm{p(t)}),
\end{split}
\end{equation}
where $p(t)=[x(t)~y(t)~z(t)]^\mathrm{T}$ is the jammer's position relative to the defender in the local Hill's frame shown in Figure \ref{fig:ref}, $G_a$ is the jammer's antenna transmission gain, $G_d(\theta(p(t)))$ is the defender's antenna reception gain that depends on the angle of reception $\theta(p(t))\in\mathbb{R}$, and $L(\norm{p(t)})$ is the free-space path loss function. For the angle $\theta(t)$, it follows from geometry and Figure \ref{fig:ref} that
\begin{equation}\label{eq:theta}
\mathrm{cos}\theta(p(t))=-\frac{x(t)}{\sqrt{x^2(t)+y^2(t)+z^2(t)}},
\end{equation}
while for the free-space path loss, we have the formula
\begin{equation}\label{eq:FSPL}
L(\norm{p(t)})=\left(\frac{\wlen}{{4\pi \norm{p(t)}}} \right)^2.
\end{equation}
In addition, the parameter $\wlen$ in \eqref{eq:FSPL} is the wavelength of the communication, which we assume to be the same both for the jammer and the defender. Hence, from \eqref{eq:SINR}-\eqref{eq:FSPL}, we notice that the effect of jamming, as quantified by the SINR, is a direct function of the relative position $p(t)$ of the jammer with respect to the defender satellite.

\begin{figure}[t]
  \centering
\begin{tikzpicture}[scale=1]

  \node at (0,0) {\includegraphics[width=6cm]{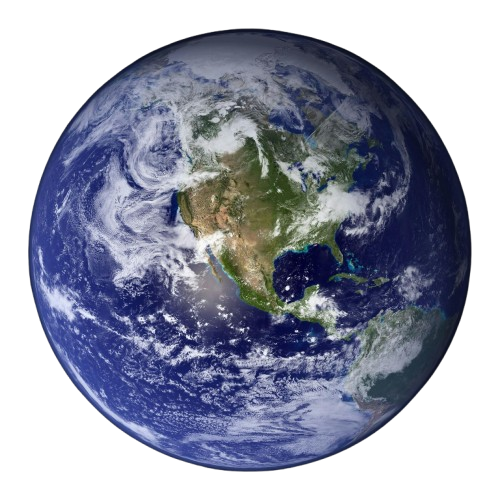}};
  
    \begin{scope}[rotate around={30:(0,0)}] 
    \draw[thick, dashed] (6,0) arc[start angle=0, end angle=360, x radius=6, y radius=1.5];
  \end{scope}


  \node at (2,2) {\includegraphics[angle=9.5, width=0.8cm]{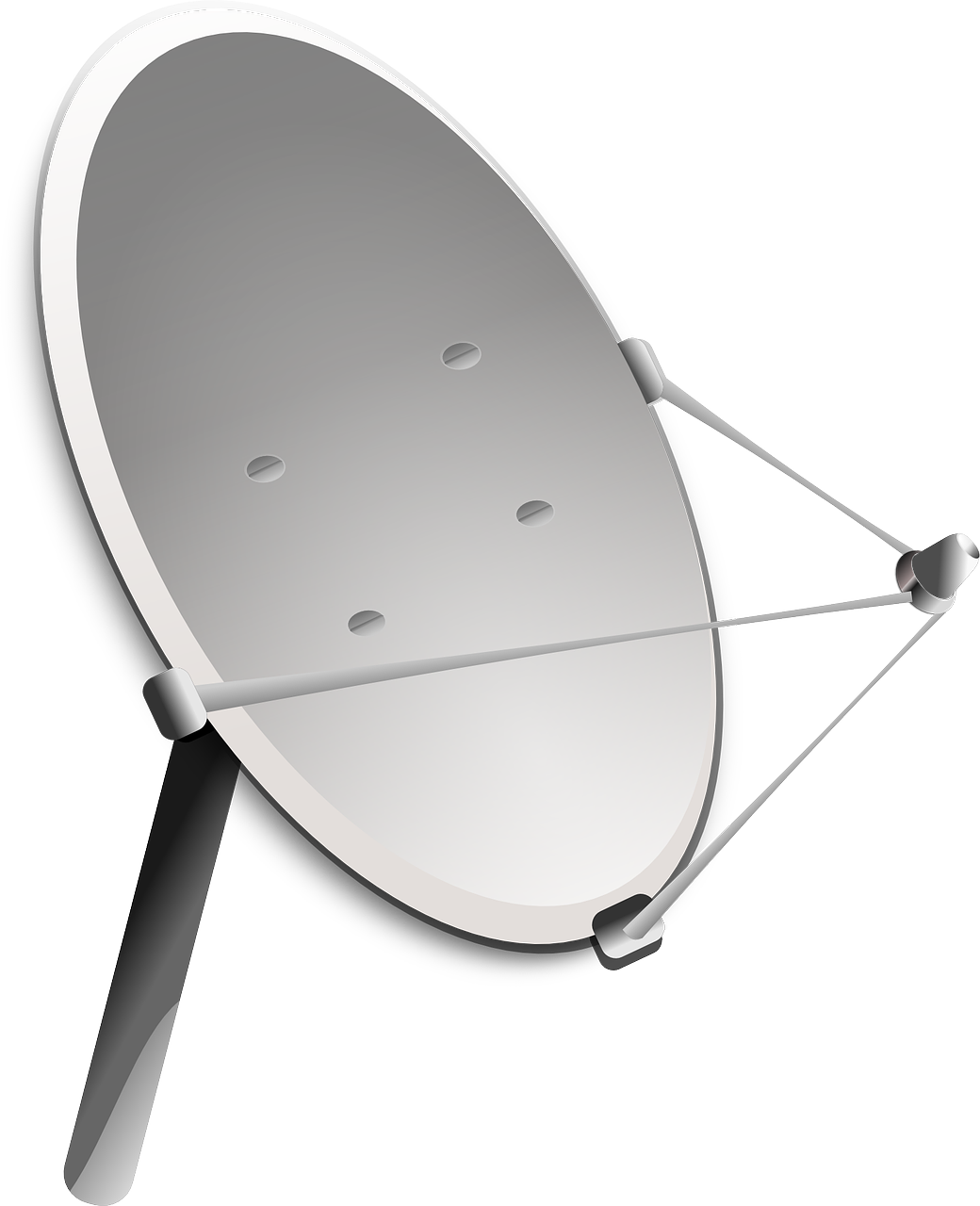}};

  \node at (5.3,3.05) {\includegraphics[angle=9.5, width=2.5cm]{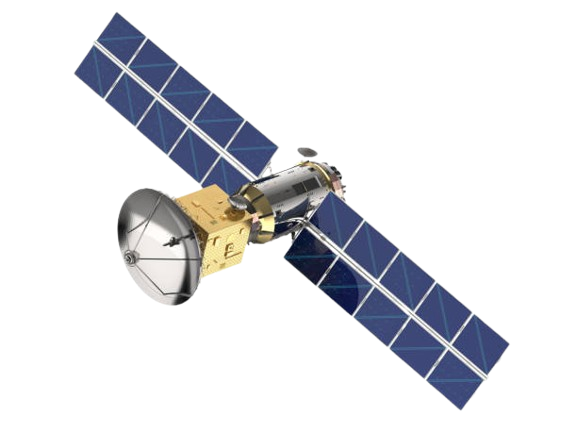}};

  \node at (3.8,0.9) {\includegraphics[angle=-135, width=2.5cm]{satellite.png}};

  \draw[->, very thick, green] (2.3,2.1) -- (4.7,2.8) node[right] {};

  \draw[->, very thick, red] (4.1,1.4) -- (4.7,2.8) node[right] {};
  \node at (3.8, 1.7) {\textcolor{red}{$a(t)$}};
  \node at (2.5,2.5) {\textcolor{green}{$s(t)$}};

  \draw[->, very thick, black] (5.17,3.07) -- ++(1,0.5) node[right] {$x$};
  \draw[->, very thick, black] (5.17,3.07) -- ++(-0.3, 1.1) node[above] {$z$};
  \draw[->, very thick, black] (5.17,3.07) -- ++(-0.8,0.8) node[above right] {$y$};

\end{tikzpicture}
\caption{The considered setup with the local Hill's frame centered on the defender. The ground station sends the friendly communication signal $s(t)$ to the defender, while the jammer disrupts this uplink communication with the interference signal $a(t)$.}\label{fig:ref}
\end{figure}

\subsection{Orbital Maneuvering Problem for Efficient Jamming}
\vspace{3mm}

Since the \textrm{SINR} \eqref{eq:SINR} is an explicit function of the position $p(t)$, 
the jammer wants to reposition itself appropriately relative to the defender's antenna in order to maximize the effect of its interference. Towards this end, consider the local reference frame centered on the defender, shown in Figure \ref{fig:ref}, with the $x$ axis pointing radially outward from the defender, the $z$ axis pointing to the orbit normal, and the $y$ axis pointing to the direction of motion. In this frame, one can describe the position of the jammer relative to the defender according to the Clohessy-Wiltshire equations:
\begin{equation}\label{eq:CW}
\begin{split}
\ddot{x}(t)&=3n^2x(t)+2n\dot{y}(t)+\frac{u_x(t)}{m},\\
\ddot{y}(t)&=-2n\dot{x}(t)+\frac{u_y(t)}{m},\\
\ddot{z}(t)&=-n^2z(t)+\frac{u_z(t)}{m}.
\end{split}
\end{equation}
Here, $t\ge0$ is time, $x, y, z\in\mathbb{R}$ are the Cartesian coordinates of the jammer in the local frame, $u_x, u_y, u_z\in\mathbb{R}$ are the external thrusts of the jammer, and $n$ is the orbital rate of the defender. Defining the full state as $w=[p^\mathrm{T}~v^\mathrm{T}]^\mathrm{T}\in\mathbb{R}^6$, where $p=[x~y~z]^\mathrm{T}$ is the position vector and $v=[\dot{x}~\dot{y}~\dot{z}]^\mathrm{T}$ the velocity vector,
and the control input as $u=[u_x~u_y~u_z]^\mathrm{T}\in\mathbb{R}^3$, one can also write \eqref{eq:CW} in the aggregate form
\begin{equation}\label{eq:CW_agg}
\dot{w}(t)=Aw(t)+Bu(t),~w(0)=w_0,
\end{equation}
where
\begin{equation}
A=\begin{bmatrix}0 & 0 & 0 & 1 & 0 & 0 \\ 0 & 0 & 0 & 0 & 1 & 0 \\ 0 & 0 & 0 & 0 & 0 & 1\\ 3n^2 & 0 & 0 & 0 & 2n & 0 \\ 0 & 0 & 0 & -2n & 0 & 0 \\ 0 & 0 & -n^2 & 0 & 0 & 0\end{bmatrix}, \qquad B=\begin{bmatrix}0 & 0 & 0 \\ 0 & 0 & 0 \\ 0 & 0 & 0 \\ \frac{1}{m} & 0 & 0 \\ 0 & \frac{1}{m} & 0 \\ 0 & 0 & \frac{1}{m}\end{bmatrix}.
\end{equation}

Given its relative orbital dynamics \eqref{eq:CW}, the purpose of the jammer is to design its thrust $u$ so that it positions itself in a way that optimizes the effect of jamming on the defender, i.e., the SINR \eqref{eq:SINR}, while also taking into account its limited fuel resources. 
\begin{problem}\label{pr:maneuver}
Find the thrust profile $u$ that i) leads the jammer to a position that maximizes interference at a target time $T>0$ while minimizing fuel; and ii) keeps the jammer in a position that maximizes interference over the interval $[T, T']$, $T'>0$, while minimizing fuel.
\end{problem}
The solution to Problem \ref{pr:maneuver} leads to a two-stage maneuver for the jammer where, in the first stage, the jammer repositions itself at an appropriate jamming location \textit{before} jamming begins and, at the second stage, the jammer keeps itself at an appropriate jamming location \textit{while} jamming is taking place. The rest of this paper focuses on mathematically characterizing and solving this problem.

\section{Thrust Control for Optimal Satellite-to-Satellite Jamming}\label{sec:main}

\subsection{Optimal Repositioning for Efficient Jamming}
\vspace{3mm}

To reposition itself for optimal jamming at a specific time $T>0$ while taking fuel limitations into account, and thus resolve the first part of Problem \ref{pr:maneuver}, the jammer may opt to solve the following optimal control problem 
\begin{equation}\label{eq:optT}
\begin{split}
&\min_{u\in\mathcal{U}} \int_0^T \frac{1}{2}u^\mathrm{T}(t) R_r u(t) \mathrm{d}t+a_r\mathrm{SINR}(p(T))\\ &\mathrm{s.t.} \quad \dot{w}(t)=Aw(t)+Bu(t),\\&~\qquad w(0)=w_0,
\end{split}
\end{equation}
where $R_r\succ0$ is a weighting matrix penalizing fuel consumption, $a_r>0$ is a weighting constant, and $\mathcal{U}$ is the space of measurable signals. Note that in \eqref{eq:optT},  we take into account the consideration of fuel limitations by penalizing the magnitude of $u$ in the running cost.

In what follows, we prove that one can obtain the optimal solution to \eqref{eq:optT} by solving a set of algebraic equations.
\begin{theorem}
Let $u^\star\in\mathcal{U}$ be the optimal solution to \eqref{eq:optT}. Then:
\begin{equation}
u^\star(t)=-R_r^{-1}B^\mathrm{T}e^{A^\mathrm{T}(T-t)}\begin{bmatrix}\mu \\ 0_3\end{bmatrix},
\end{equation}
where $p_f, v_f, \mu\in\mathbb{R}^3$ solve the equations
\begin{subequations}
\begin{align}\label{eq:bcx}
&\begin{bmatrix}p_f \\ v_f\end{bmatrix}=e^{AT}w_0-W_c\begin{bmatrix}\mu \\ 0_3\end{bmatrix},\\
& \mu=\frac{a_r P_aG_aP_sH_s(T)}{(P_aG_aG_d(\theta(p_f))L\left(\norm{p_f}\right)+\sigma_{\eta}^2)^2}\left(\left(\frac{\wlen}{4\pi}\right)^2\frac{2G_d(\theta(p_f))p_f}{\norm{p_f}^4}+\frac{L\left(\norm{p_f}\right)G_d'(\theta(p_f))}{\mathrm{sin}\theta(p_f)}\begin{bmatrix}-\frac{1}{\norm{p_f}}+\frac{x_f^2}{\norm{p_f}^3} \\ \frac{x_fy_f}{\norm{p_f}^3} \\ \frac{x_fz_f}{\norm{p_f}^3}\end{bmatrix} \right),\label{eq:bcl}
\end{align}
\end{subequations}
and where $W_c=\int_0^{T}e^{A\tau}BR_r^{-1}B^\mathrm{T}e^{A^\mathrm{T}\tau}\mathrm{d}\tau$ is the weighted controllability Gramian of $(A, B)$.
\end{theorem}
\begin{proof}
Denote the Hamiltonian of the optimal control problem \eqref{eq:optT} as
\begin{equation}
H(w, u, \lambda)=\frac{1}{2}u^\mathrm{T}R_ru+\lambda^\mathrm{T}(Aw+Bu).
\end{equation}
Note that, since $R_r\succ0$, the Hamiltonian is strictly convex in $u$. Therefore, the minimum principle implies the optimal controller should satisfy the stationarity condition $\frac{\partial H}{\partial u}=0$. From this condition, we obtain
\begin{equation}\label{eq:Hu}
\frac{\partial H}{\partial u}=0 \Longrightarrow R_ru^\star+B^\mathrm{T}\lambda=0 \Longrightarrow u^\star=-R_r^{-1}B^\mathrm{T}\lambda.
\end{equation}
In addition, the adjoint equation yields
\begin{equation}\label{eq:Hw}
\dot{\lambda}=-\frac{\partial H}{\partial w} \Longrightarrow \dot{\lambda}=-A^\mathrm{T}\lambda \Longrightarrow \lambda(t)=e^{A^\mathrm{T}(T-t)}\lambda(T).
\end{equation}
Combining \eqref{eq:Hu}-\eqref{eq:Hw}, we derive
\begin{equation}\label{eq:Hu2}
u^\star(t)=-R_r^{-1}B^\mathrm{T}e^{A^\mathrm{T}(T-t)}\lambda(T).
\end{equation}
Therefore, computing the optimal control boils down to finding the parameter $\lambda(T)$. To that end, let us denote $\lambda(T)=\begin{bmatrix}\mu^\mathrm{T} & \nu^\mathrm{T} \end{bmatrix}^\mathrm{T}$, where $\mu, \nu\in\mathbb{R}^3$, and $p(T)=p_f=\begin{bmatrix}x_f & y_f & z_f \end{bmatrix}^\mathrm{T}$. From the transversality condition, since the terminal cost in \eqref{eq:optT} is independent of the velocities, it follows that $\nu=0$. On the other hand, considering \eqref{eq:SINR}-\eqref{eq:H}, the transversality condition for $\mu$ yields
\begin{equation}\label{eq:tempmu}
\begin{split}
\mu&=\frac{\partial (a_r\mathrm{SINR}(p_f))}{\partial p_f}=\frac{\partial}{\partial p_f}\frac{a_rP_sH_s(T)}{P_aG_aG_d(\theta(p_f))L\left(\norm{p_f}\right)+\sigma_{\eta}^2}\\&=-\frac{a_rP_sH_s(T)}{(P_aG_aG_d(\theta(p_f))L\left(\norm{p_f}\right)+\sigma_{\eta}^2)^2}\frac{\partial}{\partial p_f}\left(P_aG_aG_d(\theta(p_f))L\left(\norm{p_f}\right)+\sigma_{\eta}^2\right)\\&=-\frac{a_rP_aG_aP_sH_s(T)}{(P_aG_aG_d(\theta(p_f))L\left(\norm{p_f}\right)+\sigma_{\eta}^2)^2}\bigg(G_d(\theta(p_f))\frac{\partial L\left(\norm{p_f}\right)}{\partial p_f}\\&\qquad\qquad\qquad\qquad\qquad\qquad\qquad\qquad\qquad+L\left(\norm{p_f}\right)\frac{\partial G_d(\theta(p_f))}{\partial \theta(p_f)}\frac{\partial \theta(p_f)}{\partial\mathrm{cos}\theta(p_f)}\frac{\partial \mathrm{cos}\theta(p_f)}{\partial p_f} \bigg).
\end{split}
\end{equation}
Here, using equations \eqref{eq:theta} and \eqref{eq:FSPL}, we have the identities:
\begin{equation}\label{eq:temppartials}
\begin{split}
\frac{\partial L\left(\norm{p_f}\right)}{\partial p_f}&=\frac{\partial}{\partial p_f}\left(\frac{\wlen}{{4\pi \norm{p_f}}} \right)^2=-\left(\frac{\wlen}{4\pi}\right)^2\frac{2p_f}{\norm{p_f}^4},\\
\frac{\partial \theta(p_f)}{\partial \mathrm{cos}\theta(p_f)}&=-\frac{1}{\mathrm{sin}\theta(p_f)},\\
\frac{\partial \mathrm{cos}\theta(p_f)}{\partial p_f}&=\begin{bmatrix}-\frac{1}{\norm{p_f}}+\frac{x_f^2}{\norm{p_f}^3} & \frac{x_fy_f}{\norm{p_f}^3} & \frac{x_fz_f}{\norm{p_f}^3}\end{bmatrix}^\mathrm{T}.
\end{split}
\end{equation}
Combining equations \eqref{eq:tempmu}-\eqref{eq:temppartials} yields \eqref{eq:bcl}. Finally, plugging $u=u^\star$ from \eqref{eq:Hu2} in \eqref{eq:CW_agg}, we get $\dot{w}(t)=Aw(t)-BR_r^{-1}B^\mathrm{T}e^{A^\mathrm{T}(T-t)}\lambda(T)$. Integrating this over $[0, T]$, we obtain:
\begin{equation*}
w(T)=e^{AT}w_0-\int_0^Te^{A(T-\tau)}BR_r^{-1}B^\mathrm{T}e^{A^\mathrm{T}(T-\tau)}\mathrm{d}\tau\lambda(T).
\end{equation*}
Recalling that $\lambda(T)=\begin{bmatrix}\mu^\mathrm{T} & 0_3^\mathrm{T} \end{bmatrix}^\mathrm{T}$ and denoting $w(T)=[p_f^\mathrm{T}~v_f^\mathrm{T}]^\mathrm{T}$, we obtain \eqref{eq:bcx}, concluding the proof. \frQED
\end{proof}

\begin{remark}
Without knowing the power $P_s$ of the friendly signal $s(t)$ as well as the properties of the sender's antenna and channels, 
it becomes difficult for the jammer to get an estimate of the SINR expression \eqref{eq:SINR}. Nevertheless, given a known upper bound $P$ on $P_sH_s(t)$, such that $P\ge P_s H_s(t)$, the jammer can calculate an upper bound $\overline{\mathrm{SINR}}$ of the SINR \eqref{eq:SINR} according to the formula
\begin{equation}\label{eq:SINRb}
\overline{\mathrm{SINR}}(p(t))=\frac{P}{P_aH_a(p(t))+\sigma_{\eta}^2}.
\end{equation}
Subsequently, it is straightforward to derive the conditions for solving an optimal control problem with terminal cost dictated by $\overline{\mathrm{SINR}}$, instead of $\mathrm{SINR}$.
\end{remark}

\subsection{Optimal Cruising for Efficient Jamming over a Time Window}
\vspace{3mm}

After positioning at an appropriate initial jamming location, per the second part of Problem \ref{pr:maneuver}, the jammer wants to remain in an area that allows it to jam efficiently over a communication window $[T, T']$. We capture this objective with an optimization similar to \eqref{eq:optT}, but where now the jammer wants to minimize SINR over \textit{running} time instead of at a terminal time. Specifically, we capture this objective through the optimal control problem
\begin{equation}\label{eq:optR}
\begin{split}
&\min_{u\in\mathcal{U}} \int_T^{T'} \left(\frac{1}{2}u^\mathrm{T}(t) R_c u(t)+a_c{\mathrm{SINR}}(p(t))\right) \mathrm{d}t\\ &\mathrm{s.t.} \quad \dot{w}(t)=Aw(t)+Bu(t),\\&~\qquad w(T)~\mathrm{given},
\end{split}
\end{equation}
where $R_c\succ0$ and $a_c>0$.
Unlike \eqref{eq:optT}, the solution to \eqref{eq:optR} boils down to a two-point boundary value problem, the equations of which we describe in the following result.
\begin{theorem}
Let $u^\star\in\mathcal{U}$ be the optimal solution to \eqref{eq:optR}. Then:
\begin{equation}\label{eq:ustar2}
u^\star(t)=-\frac{1}{m}R_c^{-1}\lambda_v(t),
\end{equation}
where $\lambda_p,\lambda_v:[T, T']\rightarrow\mathbb{R}^3$, so that  $\lambda=[\lambda_p^\mathrm{T}~\lambda_v^\mathrm{T}]^\mathrm{T}$ solves the equation
\begin{equation}\label{eq:bvp}
\dot{\lambda}(t)={-}A^\mathrm{T}\lambda(t)-\begin{bmatrix}\frac{a_cP_aG_aP_sH_s(t)}{(P_aG_aG_d(\theta(p(t)))L\left(\norm{p(t)}\right)+\sigma_{\eta}^2)^2}\left(\left(\frac{\wlen}{4\pi}\right)^2\frac{2G_d(\theta(p(t)))p(t)}{\norm{p(t)}^4}{+}\frac{L\left(\norm{p(t)}\right)G_d'(\theta(p(t)))}{\mathrm{sin}\theta(p(t))}\begin{bmatrix}\frac{-1}{\norm{p(t)}}{+}\frac{x^2(t)}{\norm{p(t)}^3} \\ \frac{x(t)y(t)}{\norm{p(t)}^3} \\ \frac{x(t)z(t)}{\norm{p(t)}^3}\end{bmatrix} \right) \\ 0_3\end{bmatrix}
\end{equation}
under $\lambda(T')=0$ and \eqref{eq:CW_agg} for $u=u^\star$.
\end{theorem}
\begin{proof}
Denote the Hamiltonian of the optimal control problem \eqref{eq:optR} as
\begin{equation}
H(w, u, \lambda)=\frac{1}{2}u^\mathrm{T}R_cu+a_c\mathrm{SINR}(p(t))+\lambda^\mathrm{T}(Aw+Bu).
\end{equation}
Since $R_c\succ0$, the Hamiltonian is strictly convex in $u$, and thus the optimal control must satisfy:
\begin{equation}\label{eq:Hu2b}
\frac{\partial H}{\partial u}=0 \Longrightarrow R_cu^\star+B^\mathrm{T}\lambda=0 \Longrightarrow u^\star=-R_c^{-1}B^\mathrm{T}\lambda \Longrightarrow u^\star=-\frac{1}{m}R_c^{-1}\lambda_v,
\end{equation}
which yields \eqref{eq:ustar2}. In addition, from the adjoint equation, we obtain:
\begin{equation}
\dot{\lambda}(t)=-\frac{\partial H}{\partial w}=-A^\mathrm{T}\lambda(t)-\begin{bmatrix}\frac{a_c\partial {\mathrm{SINR}}(p(t))}{\partial p(t)} \\ 0_3\end{bmatrix}.
\end{equation}
Note that we can calculate the expression for $\frac{\partial {\mathrm{SINR}}(p(t))}{\partial p(t)}$ similarly to \eqref{eq:tempmu}-\eqref{eq:temppartials}, which yields the flow equation in \eqref{eq:bvp}. Finally, the transversality condition yields $\lambda(T')=0$, which provides the boundary condition in \eqref{eq:bvp}. \frQED
\end{proof}

\begin{figure}[!t]
  \centering
  \includegraphics[width=0.76\linewidth]{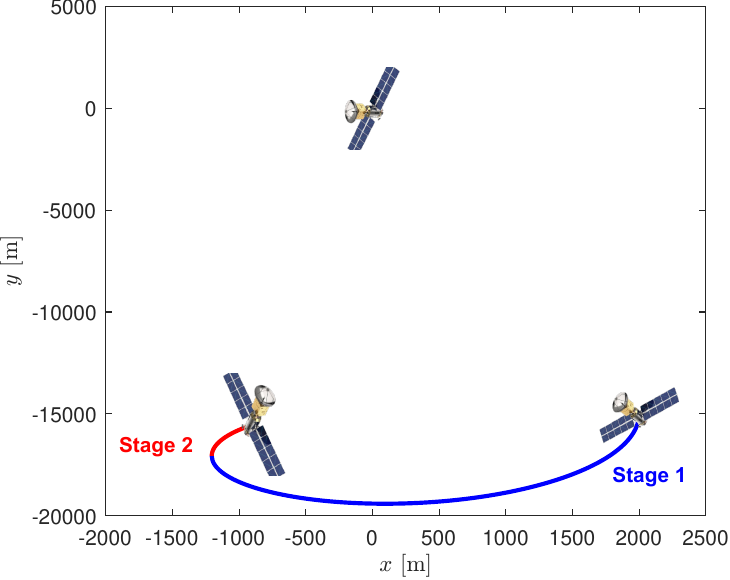}
\caption{Jammer satellite trajectory during the repositioning stage (blue) and during the jamming stage (red). The defender satellite is located at the origin of the local frame. The jammer maneuvers itself to a position from which it can view the defender's antenna. This position minimizes SINR while minimizing fuel consumption. \\ \\ \\ }
\label{fig:path}
\end{figure}

\begin{figure}[!t]
\centering
\begin{subfigure}{.49\textwidth}
  \centering
  \includegraphics[width=1\linewidth]{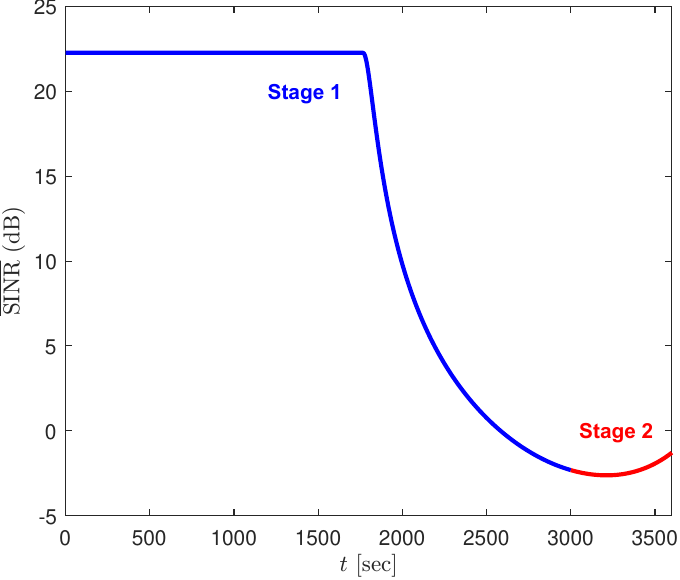}
  \caption{Trajectories of the SINR upper bound of the uplink.}
  \label{fig:sub1}
\end{subfigure}%
\hfill
\begin{subfigure}{.5\textwidth}
  \centering
  \includegraphics[width=1\linewidth]{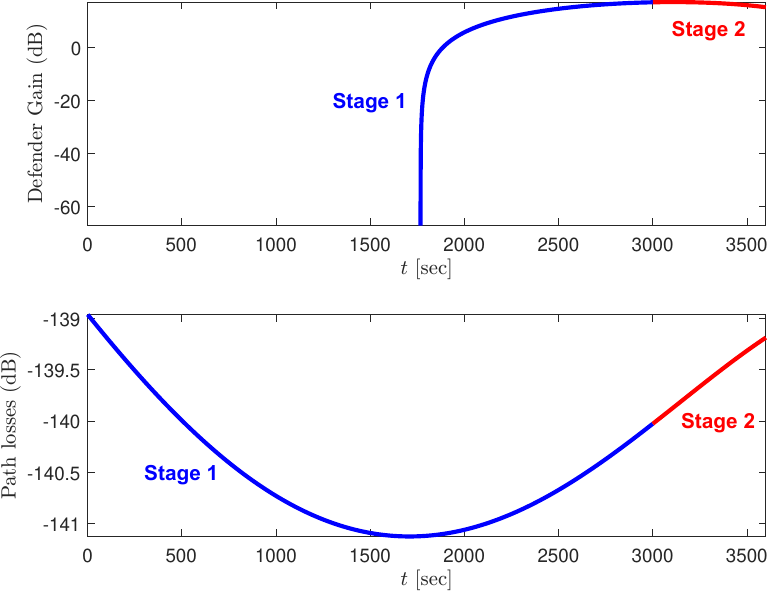}
  \caption{Trajectories of the defender's antenna gain and the free-space path losses.}
  \label{fig:sub2}
\end{subfigure}
\caption{Communication-related trajectories generated during the maneuver.}
\label{fig:coms}
\end{figure}

\begin{figure}[!t]
  \centering
  \includegraphics[width=0.8\linewidth, height=1\linewidth]{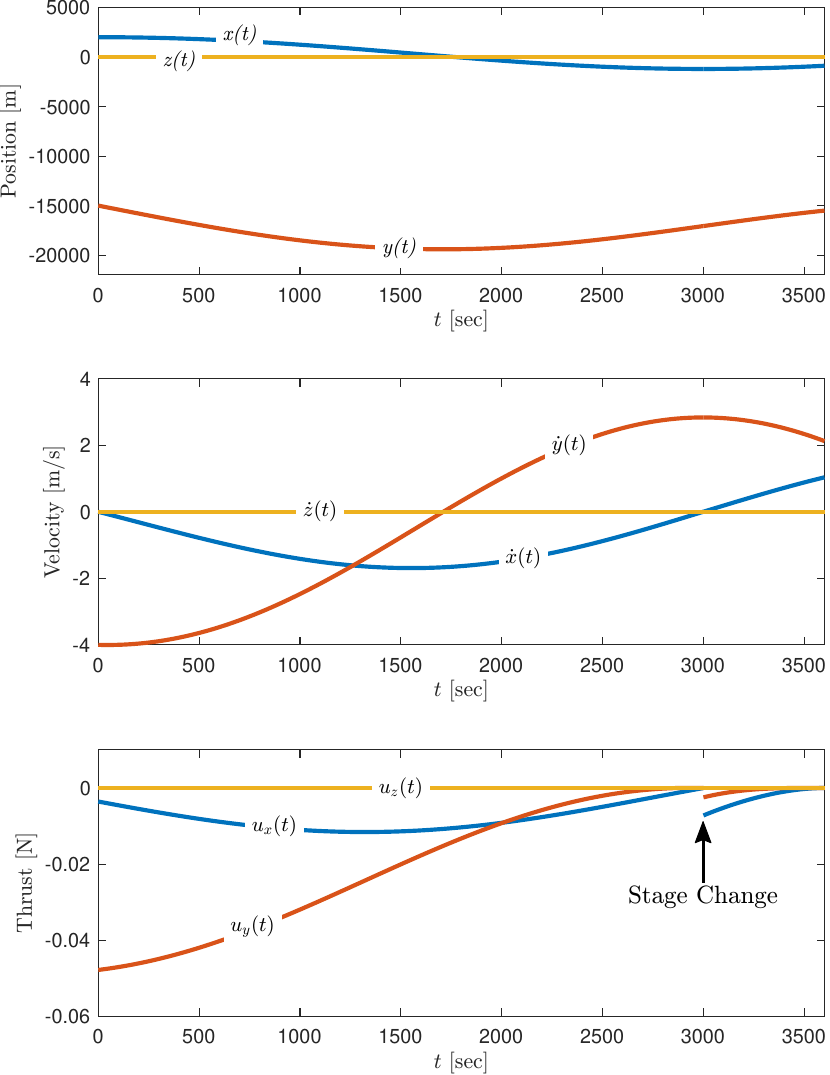}
\caption{State and control input trajectories generated during the maneuver.}
\label{fig:traj}
\end{figure}

\section{Simulation Results}\label{sec:sims}
\vspace{3mm}
We consider a defender satellite on a circular low Earth orbit at $550~\mathrm{km}$ altitude, and a jammer with a mass of $m=300~\mathrm{kg}$ in its vicinity. The jammer is positioned at $x(0)=2~\mathrm{km}$, $y(0)=-15~\mathrm{km}$, and $z(0)=0~\mathrm{km}$ in the local reference frame, with initial velocities of $\dot{x}(0)=\dot{z}(0)=0$ and $\dot{y}(0)=-4~\mathrm{m/s}$, implying that the defender and the jammer are in the same orbital plane.

The purpose of the defender is to communicate with a ground station at a frequency of $14~\mathrm{GHz}$ and a bandwidth of $500~\mathrm{MHz}$, while the purpose of the jammer is to disrupt this communication. Given Assumption \ref{ass:1}, we assume the transmission antenna gains of the ground station and the jammer are constant and equal to $30~\mathrm{dB}$, i.e., $G_s=G_a=10^3$, and that their powers are equal to $P_s=10~\mathrm{W}$ and $P_a=1~\mathrm{W}$. In addition, the antenna gain of the defender depends on the angle of signal reception $\theta$ according to the formula 
\begin{equation*}
G_d(\theta)=\begin{cases}10^4(\mathrm{cos}\theta)^2, & \theta\in(-90^\circ,~90^\circ),\\ 0, & \mathrm{otherwise}, \end{cases}
\end{equation*}
so that signals received behind the defender are assumed to be completely rejected. In addition, we assume that the thermal noise on the defender's antenna has a temperature of $250~\mathrm{K}$.

Given that communication between the ground station and the defender satellite is expected to take place in approximately $50$ to $60$ minutes, the jammer maneuvers itself to reposition to an appropriate jamming location at $T=3000~\mathrm{sec}$, and to remain there until $T'=3600~\mathrm{sec}$. To that end, it solves the optimal control problems \eqref{eq:optT} and \eqref{eq:optR} to obtain appropriate thrust profiles, with parameters chosen as $R_r=\frac{1000}{T}I_3$, $a_r=1$, $R_c=\frac{1000}{T'-T}I_3$, $a_c=\frac{100}{T'-T}$. Note that, since the uplink communication parameters are not completely known by the jammer, the jammer uses the SINR upper bound \eqref{eq:SINRb} in its formulas, with $P=P_sG_sG_d(0)L(550\cdot10^3)$. In other words, the jammer assumes the best-case communication between the defender and the ground station, wherein the defender is stationed directly above the ground station.

Figures \ref{fig:path}-\ref{fig:traj} depict the resulting maneuver. To minimize fuel consumption while maximizing jamming efficiency, we observe that the jammer opts to maneuver itself within a 90-degree angle from the defender's antenna axis, though not by a large margin. At the same time, the overall distance from the defender remained largely unchanged. 
 Meanwhile, Figure \ref{fig:coms} shows the impact of the jamming on the defender's uplink communication, where we notice that the SINR during jamming remained below $0~{\mathrm{dB}}$ -- a sign that the jamming was effective at disrupting the defender. Moreover, we observe that most of the increase in jamming efficiency was due to an optimized angle with respect to the defender's antenna, rather than a smaller distance from it.

\section{Conclusion}\label{sec:conc}
\vspace{3mm}
We consider the problem of maneuvering a satellite in low Earth orbit to efficiently jam the uplink communication of another nearby satellite.  We cast this objective as a two-stage optimal control problem, wherein the jammer maneuvers itself to minimize SINR while minimizing fuel consumption.

Future work involves developing a security mechanism that the defender satellite could use to shield itself against the proposed jamming attack.

\section*{Acknowledgments}
\vspace{3mm}
This work was sponsored in part by grants AFRL FA9550-23-1-0646 and ONR N00014-22-1-2703.

\bibliography{sample}

\end{document}